\documentclass[letterpaper, 10 pt, conference]{ieeeconf}  

\IEEEoverridecommandlockouts                              

\usepackage{graphicx} 
\usepackage[margin=1in]{geometry}
\usepackage{amsfonts}
\usepackage{amsmath}
\usepackage{amssymb}
\usepackage{mathtools} 
\usepackage{xcolor}
\usepackage{tikz}
\usetikzlibrary{arrows.meta,positioning,calc,fit}
\definecolor{plantgray}{RGB}{225,230,236}
\definecolor{methodpurple}{RGB}{205,190,255}
\definecolor{methodorange}{RGB}{255,205,135}
\definecolor{feedbackblue}{RGB}{175,215,255}
\definecolor{accentblue}{RGB}{45,95,185}
\definecolor{accentred}{RGB}{200,60,60}
\usepackage{caption}
\usepackage{array}
\usepackage{hyperref}

\usepackage{float}
\usepackage{booktabs}     

\newtheorem{definition}{Definition}
\newtheorem{theorem}{Theorem}

\newtheorem{corollary}{Corollary}
\newtheorem{lemma}{Lemma}

\newtheorem{assumption}{Assumption}

 \allowdisplaybreaks 

\title{\LARGE \bf
Sampling-Horizon Neural Operator Predictors for\\  Nonlinear Control under Delayed Inputs
}

\author{Luke Bhan$^{1, ^\ast}$, Peter Quawas$^{1}$, Miroslav Krstic$^{1}$, and Yuanyuan Shi$^{1}$
\thanks{*Corresponding author}
\thanks{$^{1}$ The authors are with the University of California San Diego, \{lbhan, pquawas, mkrstic, yyshi\}@ucsd.edu}}

\begin{document}

\maketitle
\thispagestyle{empty}
\pagestyle{empty}

\begin{abstract}
Modern control systems frequently operate under input delays and sampled state measurements. A common delay-compensation strategy is predictor feedback; however, practical implementations require solving an implicit ODE online, resulting in intractable computational cost. Moreover, predictor formulations typically assume continuously available state measurements, whereas in practice measurements may be sampled, irregular, or temporarily missing due to hardware faults. In this work, we develop two neural-operator predictor-feedback designs for nonlinear systems with delayed inputs and sampled measurements. In the first design, we introduce a sampling-horizon prediction operator that maps the current measurement and input history to the predicted state trajectory over the next sampling interval. In the second design, the neural operator approximates only the delay-compensating predictor, which is then composed with the closed-loop flow between measurements. The first approach requires uniform sampling but yields residual bounds that scale directly with the operator approximation error. In contrast, the second accommodates non-uniform, but bounded sampling schedules at the cost of amplified approximation error, revealing a practical tradeoff between sampling flexibility and approximation sensitivity for the control engineer. For both schemes, we establish semi-global practical stability with explicit neural operator error-dependent bounds. Numerical experiments on a 6-link nonlinear robotic manipulator demonstrate accurate tracking and substantial computational speedup of $25\times$ over a baseline approach. 
\end{abstract}

\section{Introduction}
In practical control systems, input delays are unavoidable. Further, measurements are often available only at discrete sampling instants due to sensing and communication bandwidth limitations. Thus, controllers must operate reliably under delayed inputs and sparse measurements while remaining computationally efficient.

In this work, we develop a predictor feedback framework for nonlinear systems with input delay and sampled measurements using neural operators. 
We propose two approaches for handling measurement updates: one assuming uniform sampling and one allowing bounded, non-uniform sampling intervals yielding two different operator feedback designs. For both instances, we establish semi-global practical stability using Lyapunov analysis, with performance bounds explicitly characterized in terms of the neural operator approximation error.

\paragraph{Predictor feedback designs}

Predictor feedback has been extensively studied for systems with input delays \cite{bekiaris2013nonlinear,fridman2014introduction}. The central idea is to compensate the delay by predicting the future system state over the delay horizon and applying the nominal feedback law to this prediction \cite{krstic2009delay}. However, the predictor is an implicit ODE requiring numerical integration $+$ successive approximations to implement. As such, numerical approaches rely on \emph{state-dependent discretizations} \cite{karafyllis2017predictor} or truncated predictor expansions requiring higher-order derivatives of the dynamics \cite{zhou2014truncated}. While these methods provide strong stability guarantees, they can incur significant and unpredictable computational cost.

More recently, neural operator–based predictors have been proposed as a data-driven alternative that learns the predictor mapping directly from trajectory data \cite{pmlr-v283-bhan25a,bhan2025stabilizationnonlinearsystemsunknown,bajraktari2026delay}. In these approaches, the predictor is approximated offline using a neural operator \cite{lu2021deeponet,li2021fourier}—a neural architecture designed to approximate mappings between function spaces—and then deployed online for fast inference. However, existing formulations assume continuously available state measurements, whereas in practice measurements are often available only at discrete sampling instants. Motivated by this gap, we study predictor feedback under sampled measurements, where state propagation between updates leads to a hybrid feedback law requiring a new operator design.

\paragraph{Notation}
We write $|x|$ for the Euclidean norm. For $v:[0,D]\to\mathbb{R}^m$, let
$\|v\|_{L^\infty[0,D]}=\sup_{x\in[0,D]}|v(x)|$, and similarly
$\|u[t]\|_{L^\infty[0,D]}:=\sup_{x\in[0,D]}|u(x,t)|$.
At sampling time $T_i$, denote the input history by
$U_i(\theta)=\{U(T_i+\theta)$; $\theta\in[-D,0)\}$.

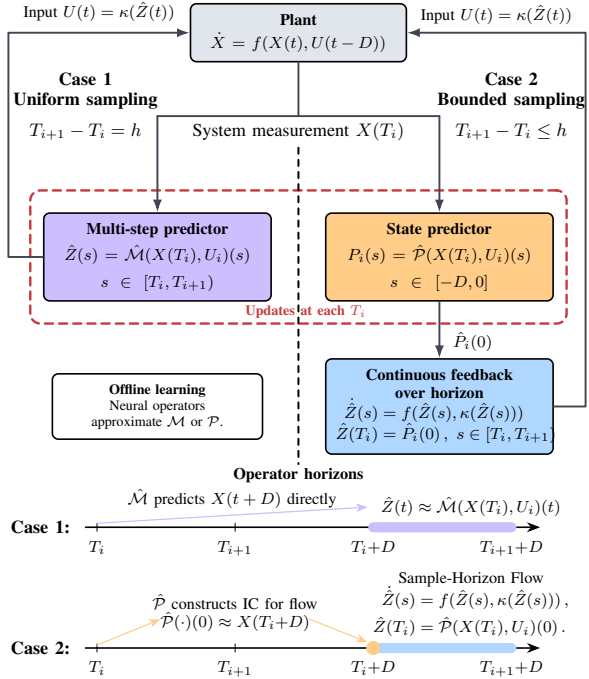
\begin{figure}[t]
\resizebox{\columnwidth}{!}{
\begin{tikzpicture}[
    >=Latex,
    font=\footnotesize,
    line join=round,
    line cap=round,
    box/.style={
        draw=black,
        rounded corners=3pt,
        thick,
        align=center,
        inner sep=5pt,
        text width=3.6cm
    },
    gbox/.style={box, fill=plantgray, text width=3.4cm},
    pbox/.style={box, fill=methodpurple},
    bbox/.style={box, fill=feedbackblue},
    obox/.style={box, fill=methodorange},
    arr/.style={-{Latex[length=2.2mm,width=1.6mm]}, line width=0.95pt,
                draw={rgb,255:red,63;green,70;blue,82}},
    line/.style={line width=0.95pt, draw={rgb,255:red,63;green,70;blue,82}}
]

\node[gbox] (plant) at (0,0) {%
    {\bfseries Plant}\\[0.08cm]
    $\dot{X}=f\!\left(X(t),U(t-D)\right)$
};

\coordinate (leftcenter)  at (-2.45,0);
\coordinate (rightcenter) at ( 2.45,0);

\node[pbox]
    (multi) at ($(leftcenter)+(0,-3.9)$) {%
    {\bfseries Multi-step predictor}\\[0.12cm]
    $\hat{Z}(s)=\hat{\mathcal{M}}(X(T_i),U_i)(s)$\\[0.15cm]
    $s\in[T_i,T_{i+1})$
};

\node[obox]
    (delay) at ($(rightcenter)+(0,-3.9)$) {%
    {\bfseries State predictor}\\[0.12cm]
    $P_i(s)=\hat{\mathcal{P}}(X(T_i),U_i)(s)$\\[0.15cm]
    $s\in[-D,0]$
};

\node[bbox]
    (feedback) at ($(rightcenter)+(0,-6.5)$) {%
    {\bfseries Continuous feedback over horizon}\\
    $\begin{aligned}
    \dot{\hat{Z}}(s) &= f(\hat{Z}(s),\kappa(\hat{Z}(s)))\\[-3pt]
    \hat{Z}(T_i) &= \hat{P}_i(0)\,,\ s\in[T_i,T_{i+1})
    \end{aligned}$
};
\node[
box,
text width=3.3cm,
font=\scriptsize
]
(operator) at ($(leftcenter)+(0,-6.5)$)
{
\textbf{Offline learning}\\
Neural operators approximate
$\mathcal{M}$ or $\mathcal{P}$.
};

\node[
draw=accentred,
dash pattern=on 4pt off 3pt,
line width=1.2pt,
rounded corners=5pt,
fit=(multi)(delay),
inner xsep=0.22cm,
inner ysep=0.38cm
] (region) {};

\node[
font=\bfseries\scriptsize,
text=accentred,
fill=white,
inner sep=1.5pt
]
at ($(region.south)+(0.13cm,0.23cm)$)
{Updates at each $T_i$};

\coordinate (split)  at (0,-1.45);
\coordinate (splitL) at (-2.45,-1.45);
\coordinate (splitR) at ( 2.45,-1.45);

\draw[line] (plant.south) -- (split);
\draw[line] (splitL) -- (splitR);
\draw[arr]  (splitL) -- (multi.north);
\draw[arr]  (splitR) -- (delay.north);

\node[font=\small] at ($(plant.south)+(0,-1.2)$) {System measurement $X(T_i)$};

\draw[dashed, line width=1.0pt]
    (0,-2) -- (0,-7.5);

\node[font=\bfseries\small, align=center, anchor=south] at (-3.7,-1.45) {Case 1\\Uniform sampling};
\node[font=\small, anchor=north] at (-3.7,-1.45) {$T_{i+1}-T_i=h$};

\node[font=\bfseries\small, align=center, anchor=south] at ( 3.7,-1.45) {Case 2\\Bounded sampling};
\node[font=\small, anchor=north] at ( 3.7,-1.45) {$T_{i+1}-T_i\leq h$};


\coordinate (lRail) at ([xshift=-0.6cm]multi.west);
\draw[arr]
    (multi.west)
    -- (lRail |- multi.west)
    -- (lRail |- plant.west)
    -- (plant.west)
    node[pos=0.5, above=2pt, font=\footnotesize] {Input $U(t) = \kappa(\hat{Z}(t))$};

\coordinate (rRail) at ([xshift=0.55cm]feedback.east);
\draw[arr]
    (feedback.east)
    -- (rRail |- feedback.east)
    -- (rRail |- plant.east)
    -- (plant.east)
    node[pos=0.5, above=0.6pt, font=\footnotesize] {Input $U(t) = \kappa(\hat{Z}(t))$};

\draw[arr] (delay.south) -- (feedback.north)
    node[pos=0.7, right=3pt, font=\footnotesize] {$\hat{P}_i(0)$};

\node[font=\bfseries\footnotesize] at (0,-7.70) {Operator horizons};

\begin{scope}[shift={(0,-7.70)}, >=Stealth, font=\small]

  \def\Ti{-3.5}
  \def\Tip{-1.1}
  \def\TiD{1.3}
  \def\TipD{3.7}

  \def\yA{-0.9}

  \draw[->, thick] (\Ti-0.2, \yA) -- (\TipD+0.5, \yA);

  \foreach \x/\lbl in {\Ti/$T_i$, \Tip/$T_{i+1}$, \TiD/$T_i{+}D$, \TipD/$T_{i+1}{+}D$} {
    \draw[thin] (\x, \yA+0.07) -- (\x, \yA-0.07);
    \node[below=3pt, font=\footnotesize] at (\x, \yA) {\lbl};
  }

\def\labelx{3} 

\draw[methodpurple, line width=5pt, line cap=round]
  (\TiD, \yA) -- (\TipD, \yA);

\node[above=5pt, black, font=\footnotesize] at (\labelx, \yA-0.1)
  {$\hat{Z}(t) \approx \hat{\mathcal{M}}(X(T_i), U_i)(t)$};

  \draw[->, semithick, methodpurple]
    (\Ti, \yA+0.07)
    to node[above, midway, black, font=\footnotesize]
         {$\hat{\mathcal{M}}$ predicts $X(t+D)$ directly}
    ({\TiD - 0.1}, \yA+0.35);

  \node[left=4pt, font=\small\bfseries] at (\Ti-0.2, \yA) {Case 1:};

  \def\yB{-3}

  \draw[->, thick] (\Ti-0.2, \yB) -- (\TipD+0.5, \yB);

  \foreach \x/\lbl in {\Ti/$T_i$, \Tip/$T_{i+1}$, \TiD/$T_i{+}D$, \TipD/$T_{i+1}{+}D$} {
    \draw[thin] (\x, \yB+0.07) -- (\x, \yB-0.07);
    \node[below=3pt, font=\footnotesize] at (\x, \yB) {\lbl};
  }

\draw[feedbackblue, line width=5pt, line cap=round]
  (\TiD, \yB) -- (\TipD, \yB);

\node[below=7pt, black, font=\footnotesize, align=center]
  at (3, \yB+1.7)
  {Sample-Horizon Flow\\
  $\begin{aligned}
  \dot{\hat{Z}}(s) &= f(\hat{Z}(s), \kappa(\hat{Z}(s)))\,,\\
  \hat{Z}(T_i) &= \hat{\mathcal{P}}(X(T_i), U_i)(0)\,.
  \end{aligned}$};

  \fill[methodorange] (\TiD, \yB) circle (3.5pt);

  \draw[->, semithick, methodorange]
    (\Ti, \yB+0.05)
    to
    (\Tip-1.3, \yB+0.5);

  \draw[->, semithick, methodorange]
    (\Tip+1.3, \yB+0.5)
    to
    (\TiD-0.1, \yB+0.1);

  \node[above, font=\footnotesize, black] at (\Tip, \yB+0.55)
    {$\hat{\mathcal{P}}$ constructs IC for flow};

  \node[above=5pt, black, font=\footnotesize] at (\Tip, \yB)
    {$\hat{\mathcal{P}}(\cdot)(0)\approx X(T_i{+}D)$};

  \node[left=4pt, font=\small\bfseries] at (\Ti-0.2, \yB) {Case 2:};

\end{scope}

\end{tikzpicture}
}

\caption{Approximate predictor-feedback implementations under uniform and bounded sampling.} 
\vspace{-1.5em}
\label{fig:predictor_feedback_diagram}

\end{figure}

\section{Problem Formulation and Control Design}
Consider the delayed-input nonlinear system
\begin{align}
    \dot{X}(t) = f(X(t), U(t-D))\,, \quad t \in [0, \infty)  \label{eq:dynamics}
\end{align}
where $X(t) \in \mathbb{R}^n$, $f: \mathbb{R}^n \times \mathbb{R}^m \to \mathbb{R}^n$, and $D>0$ is a constant, but arbitrarily long delay. State measurements are available only at discrete sampling instants $\{T_i\}_{i=0}^\infty$ with $T_0=0$. 
Further, we impose the following assumptions on the system.
\begin{assumption}
    \label{assumption:lipschitz-dynamics}
Let $f(X, U)$ be as in \eqref{eq:dynamics} and let $\mathcal{X} \subset \mathbb{R}^n$ and $\mathcal{U} \subset \mathbb{R}^m$ be compact domains with bounds in the supremum norm given by $\overline{\mathcal{X}}$ and $\overline{\mathcal{U}}$ respectively. Then, there exists a constant $C_f(\overline{\mathcal{X}}, \overline{\mathcal{U}})>0$ such that $f$ satisfies the Lipschitz condition
\begin{align}
    |f(X_1, u_1) - f(X_2, u_2)| \leq C_f(|X_1- X_2| + |u_1 - u_2|)\,, 
\end{align}
for all $X_1, X_2 \in \mathcal{X}$ and $u_1, u_2 \in \mathcal{U}$.
\end{assumption}
\begin{assumption} \label{assumption:forward-complete}
    $\dot{X} = f(X, \omega)$ in \eqref{eq:dynamics} is strongly forward complete (See \cite[Def. 6.1]{krsticDelay} for definition). 
\end{assumption}
\begin{assumption} \label{assumption:gas}
    There exists a control law $\kappa \in C^1(\mathbb{R}^n; \mathbb{R}^m)$ such that under the closed loop feedback $U(t) = \kappa(X(t))$, the delay-free system $\dot{X}(t) = f(X(t), U(t))$ is globally asymptotically stable. 
\end{assumption}
\begin{assumption}
\label{assumption:iss}
    The system $\dot{X} = f(X, \kappa(X) + \omega)$ in \eqref{eq:dynamics} is input-to-state stable (ISS). (See \cite{sontag1995characterizations} for definition). 
\end{assumption}

Assumptions \ref{assumption:lipschitz-dynamics}--\ref{assumption:iss} are standard in nonlinear delay compensation and are satisfied by a broad class of nonlinear systems including many control-affine systems and mechanical systems derived from Lagrangian models. The role of the compact sets in Assumption \ref{assumption:lipschitz-dynamics} will become clear in Section~\ref{sec:neural-operator-properties}.

To handle the delay, consider the hybrid predictor-based control law introduced by \cite{karafyllis2014numerical}
\begin{subequations}
\begin{align}
    U(s) &= \kappa(Z(s))\,,  \quad s \in [T_i, T_{i+1}) \label{eq:controlLawHybrid} \\ 
    \dot{Z}(s) &=  f(Z(s), \kappa(Z(s))) \,, \quad s \in [T_i, T_{i+1})    \label{eq:dynamicHybrid} \\ 
    Z(T_i) &= P_i(0)\,, 
    \label{eq:initHybrid}
\end{align}
\end{subequations}
where $P_i(\theta)$ is the predictor trajectory defined over the delay horizon $\theta \in [-D,0]$ for the $i$th sample measurement  
as the solution of
\begin{align}
P_i(\theta) = X(T_i) + \int_{-D}^{\theta} 
f(P_i(\tau), U(T_i+\tau))\, d\tau\,, 
\label{eq:predictorHybrid}
\end{align}

At $T_0=0$, the predictor is initialized using the preset input history for $ \theta \in [-D, 0]$ as
\begin{align}
P_0(\theta) = X(0) + \int_{-D}^{\theta} 
f(P_0(\tau), U(\tau))\, d\tau\,. 
\label{eq:predictionInit}
\end{align}

The hybrid law requires two online computations: evaluating the predictor to obtain $P_i(0)$ the initial $T_i+D$ future state estimate and integrating \eqref{eq:dynamicHybrid}--\eqref{eq:initHybrid} over each sampling interval. Both can be costly for large delays or stiff dynamics. To address this, we reformulate the hybrid feedback law as an operator and approximate it using neural operators. As illustrated in Fig. \ref{fig:predictor_feedback_diagram}, we consider two implementations: in Case~1 a neural operator directly approximates the entire predicted state $\hat{Z}$ over the sampling horizon, while in Case~2 only the predictor is approximated and used as an initial condition in \eqref{eq:dynamicHybrid}, \eqref{eq:initHybrid} which is then solved to obtain the prediction trajectory between sampling times. Case $2$ allows one to work with non-uniform measurements while case $1$ requires uniform sampling. However, as proven in Section \ref{sec:stability}, Case $2$ will require a better neural operator approximation for stability.

\section{Multi-step neural operator predictors} \label{sec:neural-operator-properties}
We begin by defining the predictor operator as the solution to the implicit differential equation in \eqref{eq:predictorHybrid}:
\begin{definition}(\textit{Predictor operator})
\label{label}
Let $X \in \mathbb{R}^n$, $U\in C^1([-D, 0]; \mathbb{R}^m)$. Then, we define the predictor operator mapping as 
$\mathcal{P}: \left(X, U\right) \rightarrow P$ where $\mathcal{P}$ maps from  $\mathbb{R}^n\times C^1([-D, 0]; \mathbb{R}^m)$ to $ C^1([-D, 0]; \mathbb{R}^n)$ and  where $P(s) = \mathcal{P}(X, U)(s)$ satisfies for all $s\in [-D, 0]$,  
\begin{alignat}{1}
    \label{eq:predictor-operator-def} P(s) - \int_{-D}^s f(P(\theta), U(\theta)) d \theta - X = 0\,.
\end{alignat}    
\end{definition}

Note, that this operator is Lipschitz following the Lemma given in \cite{pmlr-v283-bhan25a}.
\begin{lemma} (\cite[Lemma 4]{pmlr-v283-bhan25a} )\label{lemma:operator-continuity}
    Let Assumption \ref{assumption:lipschitz-dynamics} hold and let $\mathcal{X} \subset \mathbb{R}^n, \mathcal{U} \subset \mathbb{R}^m$ be bounded domains as in Assumption \ref{assumption:lipschitz-dynamics}. Then, for any $X_1, X_2 \in \mathcal{X}$ and input signals $U_1, U_2 \in C^1([-D, 0];\mathcal{U})$, the predictor operator $\mathcal{P}$ defined in \eqref{eq:predictor-operator-def} satisfies 
    \begin{alignat}{1}
    \|\mathcal{P}(X_1, U_1)&  - \mathcal{P}(X_2, U_2)\|_{L^\infty[-D, 0]} \nonumber \\  &\leq C_\mathcal{P} \left(|X_1-X_2| + \|U_1 - U_2\|_{L^\infty[-D, 0]}\right)\,,
    \end{alignat}
    with Lipschitz constant
    \begin{alignat}{1}
        C_\mathcal{P}(D, C_f) = \max(1, DC_f)e^{DC_f}\,,
    \end{alignat}
    where $D$ is the delay and $C_f$ is the system Lipschitz constant defined in Assumption \ref{assumption:lipschitz-dynamics}.  
\end{lemma}

To complete the full feedback law, consider now the continuous part of \eqref{eq:controlLawHybrid}, \eqref{eq:dynamicHybrid}:

\begin{definition} \label{def:sample-horizon-prediction-operator}
    (\textit{Sample-horizon flow operator}) Let $P \in \mathbb{R}^n$ and $h > 0$. Then, define the sample-horizon flow operator $\mathcal{Z}: P \to Z$ where $\mathcal{Z}$ maps from $\mathbb{R}^n$ to $C^{1}([0, h); \mathbb{R}^n)$ and satisfies for all $s \in [0, h)$
    \begin{align}
        Z(s) - \int_0^s f(Z(\tau), \kappa(Z(\tau))) d \tau - P  = 0\,.
    \end{align}
\end{definition}

Naturally, the sample-horizon flow operator is Lipschitz:
\begin{lemma}(\emph{Lipschitz continuity of the sample-horizon flow operator}) 
    \label{lemma:chunked-operator-lipschitz}
     Let Assumption \ref{assumption:lipschitz-dynamics}, \ref{assumption:gas} hold and let 
    $\mathcal{X} \subset \mathbb{R}^n$. Fix $h > 0$ and consider 
    $\mathcal{Z}$ defined in Definition~\ref{def:sample-horizon-prediction-operator}. Then, for any $P_1, P_2 \in \mathcal{X}$,
    the operator $\mathcal{Z}$
    satisfies
    \begin{alignat}{1}
        \|\mathcal{Z}(P_1) - \mathcal{Z}(P_2)\|_{L^\infty[0,h)}
        \le C_{\mathcal{Z}}\, |P_1 - P_2|\,,
    \end{alignat}
    where the Lipschitz constant is 
    \begin{alignat}{1}
        C_{\mathcal{Z}}(h, C_f, C_\kappa) = e^{h C_f(1+C_\kappa)}, 
    \end{alignat}
    Here $C_f$ is the Lipschitz constant of $f$ on $\mathcal{X} \times \mathcal{U}$ from 
    Assumption \ref{assumption:lipschitz-dynamics} and $C_\kappa$ is the Lipschitz constant of 
    $\kappa$ on $\mathcal{X}$ (which exists since $\kappa \in C^1$ on the bounded set $\mathcal{X}$).
\end{lemma}
The proof is given in Appendix \ref{appendix:lipschitz-proof}.

Further, the full feedback operator is given by:
\begin{definition}
    (Sampling-horizon prediction operator) 
     Let Assumption \ref{assumption:lipschitz-dynamics} hold and let $\mathcal{X} \subset \mathbb{R}^n, \mathcal{U} \subset \mathbb{R}^m$ be bounded domains as in Assumption \ref{assumption:lipschitz-dynamics} with $h > 0$. Then, define the full multi-step prediction operator as the composed mapping $\mathcal{M}: \mathbb{R}^n \times C^1([-D, 0]; \mathbb{R}^m) \to C^1([0, h); \mathbb{R}^n)$ satisfying $\mathcal{M}(X, U) = \mathcal{Z}(\mathcal{P}(X, U)(0))$.
\end{definition}

We now state our first result - namely the existence of a neural operator approximation of $\mathcal{M}$ which relies on \cite[Theorem 2.1]{lanthaler2024nonlocalitynonlinearityimpliesuniversality} and the fact that composed operator $\mathcal{M}$ is continuous given Lemmas \ref{lemma:operator-continuity} and \ref{lemma:chunked-operator-lipschitz}. 
\begin{lemma} \label{lemma:approximationNeuralOperator}
    Let $X \in \mathcal{X} \subset \mathbb{R}^n$ and $U \in C^1([-D, 0]; \mathcal{U})$
    where $\mathcal{X}$ and $\mathcal{U} \subset \mathbb{R}^m$ are bounded domains as in Assumption \ref{assumption:lipschitz-dynamics}. Additionally, let $h >0$.
    Fix a compact set $K \subset \mathcal{X} \times C^1([-D, 0]; \mathcal{U})$. Then, for all $\overline{\mathcal{X}}, \overline{\mathcal{U}}$,  $\epsilon > 0$, there exists a neural operator approximation $\hat{\mathcal{M}}: K \rightarrow C^1([0, h); \mathbb{R}^n)$ such that $\forall \theta \in [0, h)$
    \begin{equation}
        \sup_{(X, U) \in K }|\mathcal{M}(X, U)(\theta) - \hat{\mathcal{M}}(X,U)(\theta)| < \epsilon\,.
    \end{equation}
\end{lemma}

Notice, the lemma guarantees the \emph{existence} of a neural operator achieving accuracy $\epsilon$, but does not specify an explicit architecture; see \cite[Prop. 1]{luke} for further discussion. In practice, the engineer can use their neural network testing error as an estimate for $\epsilon$.  Furthermore, $\mathcal{M}$, by definition, requires uniformly spaced measurements ($T_{i+1}-T_i=h$). As mentioned, Case~2 addresses this setting by approximating the predictor operator at each sampling instant and propagating the dynamics between updates. The following estimate will be useful in this case.

\begin{corollary} \label{corr:one} (To Lemma \ref{lemma:chunked-operator-lipschitz})
Let $\mathcal X\subset\mathbb R^n$ and $\mathcal U\subset\mathbb R^m$
be bounded sets as in Assumption~\ref{assumption:lipschitz-dynamics}.
Suppose the predictor operator approximation satisfies $\forall \theta\in[-D,0], K_1$ compact
\begin{align}
\sup_{(X,U)\in K_1 \subset \mathcal{X} \times \mathcal{U}}
|\mathcal P(X,U)(\theta)-\hat{\mathcal P}(X,U)(\theta)| \label{eq:uat-predictor}
\leq \epsilon\,. 
\end{align}
Let $Z(\cdot)$ and $\hat Z(\cdot)$ denote the solutions of
\eqref{eq:dynamicHybrid}--\eqref{eq:initHybrid} with initializations
$Z(T_i)=\mathcal P(X(T_i),U_i)(0)$ and
$\hat Z(T_i)=\hat{\mathcal P}(X(T_i),U_i)(0)$, respectively.
Then for any sampling interval satisfying $T_{i+1}-T_i\le h$,
\begin{align}
|\hat Z(s)-Z(s)|
\le
\epsilon e^{hC_f(1+C_\kappa)},
\quad 
\forall s\in[T_i,T_{i+1})\,. 
\end{align}
\end{corollary}
The existence of a predictor satisfying \eqref{eq:uat-predictor} is a direct consequence of Lemma \ref{lemma:operator-continuity} and \cite[Theorem 2.1]{lanthaler2024nonlocalitynonlinearityimpliesuniversality}.  


\section{Stability Analysis} \label{sec:stability}
We now present our main result: closed-loop semi-global practical stability under two neural operator predictor feedback schemes. In the first scheme, we approximate the entire sampling-horizon predictor operator $\mathcal{M}$, which requires uniform sampling. In the second scheme, we approximate the predictor operator $\mathcal{P}$ and propagate the resulting state through the continuous ODE describing $\hat{Z}$. Hence, the second approach requires a stronger approximation of the operator, but allows for nonuniform sampling times subject to a uniform bound. In both cases, stability is achieved with a residual bound depending on the neural operator approximation error $\epsilon$.

\noindent \textbf{Case 1 (Uniform sampling, sampling-horizon prediction operator approximation).}
If $T_{i+1}-T_i=h$ and $\hat Z$ is generated using a neural operator
approximation of the multi-step prediction operator
\begin{align}
U(T_i+s) &= \kappa(\hat{Z}(T_i+s)), \qquad s \in [0, h), \\
\hat{Z}(T_i+s) &= \hat{\mathcal{M}}(X(T_i), U_i)(s),
\end{align}
for all $i$, then define the error by the function:
\begin{align}
\delta(\epsilon)=\epsilon\,. 
\end{align}

\noindent
\textbf{Case 2 (Bounded sampling, predictor approximation).}
If $T_{i+1}-T_i\le h$ and $\hat Z$ is generated by approximating
the predictor operator $\mathcal P$ at each sampling instant and
propagating the ODE as in Corollary \ref{corr:one}:
\begin{subequations}
\begin{align}
U(s) &= \kappa(\hat{Z}(s)), \qquad s \in [T_i, T_{i+1}), \\
\dot{\hat{Z}}(s) &= f(\hat{Z}(s), \kappa(\hat{Z}(s)))\,, \\ 
\hat{Z}(T_i) &= \hat{\mathcal{P}}(X(T_i), U_i)(0)\,, 
\end{align}
\end{subequations}
for all $i$, then define the error function:
\begin{align}
    \delta(\epsilon) = \epsilon e^{C_f(1+C_\kappa)h}\,. 
\end{align}

Then, in both cases, with the respective choice of $\delta(\epsilon)$, the following theorem holds:
\begin{theorem} \label{thm:main-result}
Suppose Assumptions~\ref{assumption:lipschitz-dynamics}-\ref{assumption:iss} hold, and let
$\kappa$ denote the stabilizing feedback from
Assumption~\ref{assumption:gas} with sampling times $\{T_i\}_{i=0}^\infty$ of $\mathbb{R}_+$.  Let $R>0$ be given.
Then there exist functions $\alpha^\ast\in\mathcal K_\infty$
and $\beta^\ast\in\mathcal{KL}, \epsilon^\ast(R) > 0$ such that the following holds.
If
\begin{align}
\epsilon < \epsilon^\ast(R),
\qquad
\epsilon^\ast(R)\coloneq
\delta^{-1}\big((\alpha^\ast)^{-1}(R)\big), \label{eq:epstar}
\end{align}
and the initial condition satisfies
\begin{align}
|X(0)|+& \sup_{-D\le \theta\le0}|U(\theta)| \leq \Omega(\epsilon,R),  \\ 
\Omega(\epsilon,R)&:= \bar{\alpha}^{-1}\big(R-\alpha^\ast(\delta(\epsilon))\big)\,, \quad \bar{\alpha}(r) \coloneq \beta^\ast(r,0)\,, 
\end{align}
then the resulting closed-loop system satisfies $\forall t \geq 0$
\begin{align}
|X(t)|+&\sup_{t-D\le \theta\le t}|U(\theta)|
\nonumber \\ \le&
\beta^\ast\left(
|X(0)|+\sup_{-D\le \theta\le0}|U(\theta)|, t
\right) + 
\alpha^\ast(\delta(\epsilon)),
\end{align}
Hence the closed-loop system is semi-globally practically stable
with residual bound $\alpha^\ast(\delta(\epsilon))$.
\end{theorem}

The proof is available in Appendix \ref{appendix:main-result-proof}.
Theorem~\ref{thm:main-result} establishes that both neural operator implementations preserve the stabilizing properties of the nominal predictor feedback law up to an approximation error determined by $\epsilon$. In particular, the result highlights a trade-off for the engineer: as the approximation error decreases, the admissible region of attraction increases. However, achieving smaller approximation errors on a larger region typically requires more accurate models or larger training data, making sizable regions of attraction harder to guarantee in practice. The result also highlights a structural difference between the two approaches: in Case~1 the residual bound scales directly with the operator approximation error, whereas in Case~2 the bound is amplified by the propagation factor $e^{C_f(1+C_\kappa)h}$ due to the continuous ODE evolution between sampling instants, but makes no assumption of uniform sampling. Hence, the neural network design must account for the available sampling guarantees: weaker sampling conditions require more accurate neural network approximations to maintain the stability certificate.

\section{Numerical Results}
\begin{figure*}
    \centering
    \includegraphics{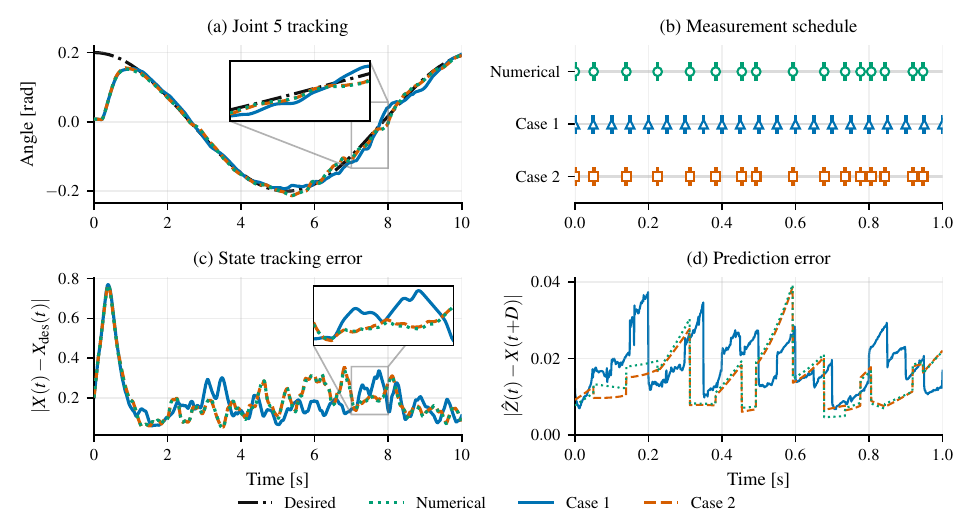}
    \caption{(a) Joint~5 tracking with delay $D=0.2$s.
(b) Measurement sampling schedules: numerical and Case~2 use randomized intervals in $[0.02,0.1]$, while Case~1 uses fixed sampling $h=0.05$. Measurements include Gaussian noise $\mathcal{N}(0,0.01)$. 
(c) Moving average ($0.25$s) of the $L^2$ tracking error.
(d) Prediction error - zoomed to the first $1$s.}
    \label{fig:mainfig}
    \vspace{-1em}
\end{figure*}
All code/data is available at \cite{bhan2026neuraloperatorpredictors}. 
We consider delayed torque control of a telerobotic manipulator:
\begin{align}
    M(X(t))\ddot{X}(t) + C(X(t), \dot{X}(t)) + G(X(t)) = \tau(t-D),
\end{align}
where $M$, $C$, and $G$ denote the inertia, Coriolis, and gravitational terms, $X(t)\in\mathbb{R}^{n_q}$ is the vector of joint angles, and $\tau$ is the applied torque. The delay is induced by the telerobotic signal time and note the dynamics are a challenging system in which we consider an xArm6 (6-link) robot leading to a full state size of $[X, \dot{X}] \in \mathbb{R}^{12}$ (position and velocity for each link).

 We consider tracking the reference trajectory
\begin{align}
X_{\mathrm{des}}(t) = X_0 + A\sin(\omega t + \phi), \nonumber 
\end{align}
where $X_0$ is a starting configuration, $A,\omega>0$ are the amplitude and frequency, and $\phi\in\mathbb{R}^{6}$ are phase offsets.

For both approaches, we use the Fourier Neural Operator \cite{li2021fourier} (which satisfies Lemma \ref{lemma:approximationNeuralOperator}) and generate $20$k training pairs by rolling out noisy trajectories with a numerical predictor solved to machine precision under the standard computed-torque controller \cite{spong2006robot}. Both operators achieve an $L^2$ validation error of $10^{-4}$ and full hyperparameters are provided in \cite{bhan2026neuraloperatorpredictors}. 

In Fig. \ref{fig:mainfig}, we highlight an example trajectory of our system where system measurements are sampled under Gaussian noise $\mathcal{N}(0, 0.01)$. We consider three feedback approaches: a neural operator approximation over a fixed horizon as in Case 1, a numerical implementation of the predictor \cite{karafyllis2014numerical} as in case 2, and a neural operator approximation of the predictor as in case 2. In all three approaches, the system successfully tracks the desired trajectory; furthermore, note in panel (c), how the resets of measurements directly affect the predictor error. Finally, consistent with prior work \cite{bhan2025stabilizationnonlinearsystemsunknown, bajraktari2026delay}, the neural operator approximations provide a significant computational advantage. Using GPU acceleration (Nvidia~4070), the neural operators require approximately $1\,\mathrm{ms}$ per evaluation, whereas the numerical predictor requires approximately $25\,\mathrm{ms}$ when solved to a tolerance of $10^{-6}$.
\section{Conclusion}
We introduced two neural-operator predictor–feedback architectures for nonlinear systems with delayed inputs and sampled measurements. For both cases, we showed semi-global practical stability estimates whose robustness directly depends on sampling assumptions and the neural operator training error. Lastly, numerical experiments on a $6$-link robotic manipulator demonstrate the computational advantages of the approach.
\appendix
\section{Appendix}
\subsection{Proof of Lemma \ref{lemma:chunked-operator-lipschitz}}\label{appendix:lipschitz-proof}
\begin{proof}
    Let $P_1, P_2 \in \mathcal{X}$ and denote $Z_1(\cdot) = \mathcal{Z}(P_1)(\cdot)$ and likewise $Z_2(\cdot) = \mathcal{Z}(P_2)(\cdot)$. Then, we have that $\forall s \in [0, h)$
    \begin{align}
        |\mathcal{Z}(P_1)(s) - \mathcal{Z}&(P_2)(s)| \nonumber \\\leq & |P_1-P_2|  \nonumber + \int_0^s |f(Z_1(\tau), \kappa(Z_1(\tau))) \nonumber \\ 
        & - f(Z_2(\tau), \kappa(Z_2(\tau)))|\,d\tau.
    \end{align}
    Using the Lipschitz properties of $\kappa$ and $f$, we have that 
    \begin{align}
          |\mathcal{Z}(P_1)(s) - \mathcal{Z}(P_2)(s)| \nonumber \leq& |P_1-P_2|  + C_f(1+C_\kappa) \\ & \times \int_0^s |Z_1(\tau) - Z_2(\tau)| d\tau \,.
    \end{align}
    Applying Gronwall's yields $\forall s \in [0, h)$\,, 
    \begin{equation}
           |\mathcal{Z}(P_1)(s) - \mathcal{Z}(P_2)(s)| \leq |P_1-P_2|e^{s(C_f(1+C_\kappa))}\,.
    \end{equation}
    The supremum over $s\in[0,h)$ completes the result. 
\end{proof}

\subsection{Proof of Theorem \ref{thm:main-result}}\label{appendix:main-result-proof}
\begin{proof}
First, the hybrid feedback law is well-posed following \cite[Claim 1]{Karafyllis2014}. From here, the proof has three parts: (a) a transport PDE estimate for delay compensation (b) a Lyapunov-based estimate on $X(t)$, and (c) bounding  $|\hat{Z}(t) - X(t+D)|$ for each case.
    
    \noindent \underline{\textbf{(a) A bound on the transport PDE}}
    System \eqref{eq:dynamics} is equivalent to
    \begin{subequations}
    \begin{align}
        \dot{X}(t) =& f(X(t), u(0, t))\,, \\
        u_t(x, t) =& u_x(x, t) \quad x \in [0, D)\,, \\
        u(D, t) =& U(t)\,, 
    \end{align}
    \end{subequations}
    Then, using the backstepping transform
    \begin{align}
        w(x, t) = u(x, t) - \kappa(X(t+x))\,, 
    \end{align}
    we obtain the system
    \begin{subequations}
    \begin{align}
        \dot{X}(t) =& f(X(t), \kappa(X(t)) + w(0, t))\,, \\
        w_t(x, t) =& w_x(x, t) \quad x \in [0, D) \,, \label{eq:target-pde-1} \\
        w(D, t) =& \kappa(\hat{Z}(t)) - \kappa(X(t+D)) \label{eq:target-pde-2}\,. 
    \end{align}
    \end{subequations}
   It follows from \cite{bhan2025neuraloperatorspredictorfeedback} that the transport PDE system \eqref{eq:target-pde-1}–\eqref{eq:target-pde-2} satisfies
    \begin{align}
        \|w[t]\|_{L^\infty[0, D]} \leq& \|w[0]\|_{L^\infty[0, D]}e^{c(D-t)} \nonumber \\ &+ e^{cD}
 \sup_{0 \leq \tau \leq t }|w(D, \tau)|   \,, 
    \end{align}
    for any $c > 0$. Using $\kappa \in C^1(\cdot)$ on the bounded set, we have
\begin{align}
    |w(D,t)|
    \leq C_\kappa |\hat Z(t)-X(t+D)|.
\end{align}
Therefore, the transport estimate yields
\begin{align}
    \|w[t]\|_{L^\infty[0,D]}
    \leq& \nonumber 
    \|w[0]\|_{L^\infty[0,D]} e^{c(D-t)}+ e^{cD} C_\kappa \\ &\times  \sup_{0 \leq \tau \leq t} |\hat{Z}(\tau) - X(\tau+D)|, \label{eq:w-practical-estimate}
\end{align}

\noindent \textbf{\underline{(b) Stability of the closed-loop system}}
By Assumption \ref{assumption:iss}, there exist a Lyapunov function $\breve{V}(X(t))$ that is proper and radially unbounded for the delay-free system and class $\mathcal{K}_\infty$ functions $\alpha_1,\alpha_2,\alpha_3,\alpha_4$ such that
\begin{align}
\alpha_1(|X(t)|) \le \breve{V}(X(t)) \le \alpha_2(|X(t)|), \label{eq:lyapunov-iss-assumption}
\end{align}
and
\begin{align}
\frac{\partial \breve{V}(X(t))}{\partial X}& f\big(X(t),\kappa(X(t))+w(0,t)\big)
\nonumber \\ &\le -\alpha_3(|X(t)|)+\alpha_4(|w(0,t)|).  
\end{align}
Using the estimate in \eqref{eq:w-practical-estimate}, we obtain 
\begin{align}
    \dot{\breve{V}} \leq& -\alpha_3(|X(t)|) + \alpha_4\big(\|w[0]\|_{L^\infty[0, D]}e^{c(D-t)} \nonumber   \\ &\quad + e^{cD} C_\kappa \sup_{0 \leq \tau \leq t} |\hat{Z}(\tau) - X(\tau+D)|\big)\,, 
\end{align}
which using \cite[Lemma C.3]{kkk}, implies that there exists $\beta_1 \in \mathcal{KL}$, $\alpha_5 \in \mathcal{K}$ 
\begin{align}
    \breve{V}(t) \nonumber  \leq& \beta_1(|\breve V(0)|, t) +\alpha_5\big(\|w[0]\|_{L^\infty[0, D]}e^{c(D-t)} \\ & + e^{cD} C_\kappa \sup_{0 \leq \tau \leq t} |\hat{Z}(\tau) - X(\tau+D)|\big)\,, 
\end{align}
Using properties of $\mathcal{K}, \mathcal{KL}$ functions, we have
\begin{align}
    \breve{V}(t) \leq& \beta_2(|X(0)| + \|w[0]\|_{L^\infty[0, D]}, t) \nonumber \\ & +\alpha_5(2 e^{cD}C_\kappa \sup_{0 \leq \tau \leq t} |\hat{Z}(\tau) - X(\tau+D)|\big)\,, \label{eq:vhat-bound}
\end{align}
with $\beta_2 \in \mathcal{KL}$. To complete the proof in terms of an estimate on $u$, notice that the backstepping transform is invertible and hence there exists $\alpha_6, \alpha_7 \in \mathcal{K}$ such that
\begin{align}
    |X(t)| + \|w[t]\|_{L^{\infty}[0, D]} \leq& \alpha_{6}\left(|X(t)| + \|u[t]\|_{L^\infty[0, D]}\right)\,, \label{eq:w-u} \\ 
    |X(t)| + \|u[t]\|_{L^{\infty}[0, D]} \leq& \alpha_{7} \left(|X(t)| + \|w[t]\|_{L^{\infty}[0, D]}\right)\,.  \label{eq:u-w}
\end{align}
Hence, using the bounds  \eqref{eq:w-practical-estimate}, \eqref{eq:lyapunov-iss-assumption}, \eqref{eq:vhat-bound}, \eqref{eq:w-u}, \eqref{eq:u-w}, and properties of class $\mathcal{K}$ functions, we obtain
\begin{align}
    \label{eq:main-result} |X(t)| &+ \|u[t]\|_{L^\infty[0,D]} \nonumber \\ \leq& \beta^\ast (|X(0)| + \|u[0]\|_{L^\infty[0,D]}, t) \nonumber \\ & + \alpha^\ast (\sup_{0 \leq \tau \leq t}|\hat{Z}(\tau) - X(\tau+D)|)\,, \\ 
    \beta^\ast(s, t) \coloneq&\; \alpha_7\!\left(2\alpha_1^{-1}(2\beta_2(\alpha_6(s), t))\right)
    \nonumber \\ &+ \alpha_7\!\left(2\alpha_6(s)e^{c(D-t)}\right), \\
    \alpha^\ast(s)\coloneq&\; \alpha_7\!\left(2\alpha_1^{-1}(2 \alpha_5(2e^{cD}C_\kappa s))\right)
   \nonumber  \\ &+ \alpha_7\!\left(2e^{cD}C_\kappa s\right). 
\end{align}

\noindent \textbf{\underline{(c) Bounding the predicted state difference}}
The final estimate for $\delta$ is obtained in two cases:

(i) Case 1: Fixed sample measurements. 
For all $t \in [T_i,T_{i+1})$,
\begin{align}
|\hat{Z}(t)-X(t+D)|
&\le \epsilon,
\end{align}
since under the \emph{exact} predictor feedback law $Z(t)=X(t+D)$\cite{karafyllis2014numerical} and Lemma \ref{lemma:approximationNeuralOperator} holds. Because the approximation bound also holds at $T_i$ and $Z(T_i)=X(T_i+D)$ exactly, recomputing $\hat Z$ at sampling times preserves the $\epsilon$ bound. Hence the estimate holds for all $t\ge 0$.

(ii) Case 2: Nonuniform sampling with uniformly bounded inter-sample times. Since only the predictor is approximated, from Corollary \ref{corr:one} and the fact $Z(t) = X(t+D)$ under \emph{exact} implementation, we have for all $t \in [T_i, T_{i+1})$,
\begin{align}
    |\hat{Z}(t) - X(t+D)| \leq \epsilon e^{C_f(1+C_\kappa)h}\,. 
\end{align}
Note that since the approximation holds at $T_i$, the reset jump for each sampled measurement will be a maximum of size $\epsilon$ and hence the result holds. 

Lastly, notice that the neural operator approximation theorem's hold only on compact sets. Substituting \eqref{eq:main-result} and using the definitions of $\Omega$ and $\epsilon^\ast$, the stability bound implies that all trajectories remain in the compact set determined by $R$, and hence the approximation bounds remain valid for all $t\ge0$ by a standard continuation argument, completing the semi-global result.
\end{proof}

\bibliography{references}

\begin{thebibliography}{10}
\providecommand{\url}[1]{#1}
\csname url@rmstyle\endcsname
\providecommand{\newblock}{\relax}
\providecommand{\bibinfo}[2]{#2}
\providecommand\BIBentrySTDinterwordspacing{\spaceskip=0pt\relax}
\providecommand\BIBentryALTinterwordstretchfactor{4}
\providecommand\BIBentryALTinterwordspacing{\spaceskip=\fontdimen2\font plus
\BIBentryALTinterwordstretchfactor\fontdimen3\font minus \fontdimen4\font\relax}
\providecommand\BIBforeignlanguage[2]{{%
\expandafter\ifx\csname l@#1\endcsname\relax
\typeout{** WARNING: IEEEtran.bst: No hyphenation pattern has been}%
\typeout{** loaded for the language `#1'. Using the pattern for}%
\typeout{** the default language instead.}%
\else
\language=\csname l@#1\endcsname
\fi
#2}}

\bibitem{bekiaris2013nonlinear}
N.~Bekiaris-Liberis and M.~Krstic, \emph{Nonlinear Control Under Nonconstant Delays}.\hskip 1em plus 0.5em minus 0.4em\relax Philadelphia: SIAM, 2013.

\bibitem{fridman2014introduction}
E.~Fridman, \emph{Introduction to Time-Delay Systems: Analysis and Control}, ser. Systems \& Control: Foundations \& Applications.\hskip 1em plus 0.5em minus 0.4em\relax Cham: Birkhäuser / Springer, 2014.

\bibitem{krstic2009delay}
M.~Krstic, \emph{Delay Compensation for Nonlinear, Adaptive, and PDE Systems}.\hskip 1em plus 0.5em minus 0.4em\relax Boston: Birkhäuser, 2009.

\bibitem{karafyllis2017predictor}
I.~Karafyllis and M.~Krstic, \emph{Predictor Feedback for Delay Systems: Implementations and Approximations}, ser. Systems \& Control: Foundations \& Applications.\hskip 1em plus 0.5em minus 0.4em\relax Cham: Birkh{\"a}user, 2017.

\bibitem{zhou2014truncated}
B.~Zhou, \emph{Truncated Predictor Feedback for Time-Delay Systems}, ser. Lecture Notes in Control and Information Sciences.\hskip 1em plus 0.5em minus 0.4em\relax Berlin / Heidelberg: Springer, 2014.

\bibitem{pmlr-v283-bhan25a}
\BIBentryALTinterwordspacing
L.~Bhan, P.~Qin, M.~Krstic, and Y.~Shi, ``Neural operators for predictor feedback control of nonlinear delay systems,'' in \emph{Proceedings of the 7th Annual Learning for Dynamics \&amp; Control Conference}, ser. Proceedings of Machine Learning Research, N.~Ozay, L.~Balzano, D.~Panagou, and A.~Abate, Eds., vol. 283.\hskip 1em plus 0.5em minus 0.4em\relax PMLR, 04--06 Jun 2025, pp. 179--193. [Online]. Available: \url{https://proceedings.mlr.press/v283/bhan25a.html}
\BIBentrySTDinterwordspacing

\bibitem{bhan2025stabilizationnonlinearsystemsunknown}
\BIBentryALTinterwordspacing
L.~Bhan, M.~Krstic, and Y.~Shi, ``Stabilization of nonlinear systems with unknown delays via delay-adaptive neural operator approximate predictors,'' 2025. [Online]. Available: \url{https://arxiv.org/abs/2509.26443}
\BIBentrySTDinterwordspacing

\bibitem{bajraktari2026delay}
F.~Bajraktari, L.~Bhan, M.~Krstic, and Y.~Shi, ``Delay compensation of multi-input distinct delay nonlinear systems via neural operators,'' in \emph{Proceedings of the American Control Conference (ACC)}, 2026, to appear.

\bibitem{lu2021deeponet}
L.~Lu, P.~Jin, G.~Pang, Z.~Zhang, and G.~E. Karniadakis, ``Learning nonlinear operators via {DeepONet} based on the universal approximation theorem of operators,'' \emph{Nature Machine Intelligence}, vol.~3, pp. 218--229, 2021.

\bibitem{li2021fourier}
Z.~Li, N.~Kovachki, K.~Azizzadenesheli, B.~Liu, K.~Bhattacharya, A.~Stuart, and A.~Anandkumar, ``Fourier neural operator for parametric partial differential equations,'' \emph{International Conference on Learning Representations}, 2021.

\bibitem{krsticDelay}
M.~Krstic, ``Input delay compensation for forward complete and strict-feedforward nonlinear systems,'' \emph{IEEE Transactions on Automatic Control}, vol.~55, no.~2, pp. 287--303, 2010.

\bibitem{sontag1995characterizations}
E.~D. Sontag and Y.~Wang, ``On characterizations of the input-to-state stability property,'' \emph{Systems \& Control Letters}, vol.~24, no.~5, pp. 351--359, 1995.

\bibitem{karafyllis2014numerical}
I.~Karafyllis and M.~Krstic, ``Numerical schemes for nonlinear predictor feedback,'' \emph{Mathematics of Control, Signals, and Systems}, vol.~26, no.~4, pp. 519--546, 2014.

\bibitem{lanthaler2024nonlocalitynonlinearityimpliesuniversality}
\BIBentryALTinterwordspacing
S.~Lanthaler, Z.~Li, and A.~M. Stuart, ``{Nonlocality and Nonlinearity Implies Universality in Operator Learning},'' 2024. [Online]. Available: \url{https://arxiv.org/abs/2304.13221}
\BIBentrySTDinterwordspacing

\bibitem{luke}
L.~Bhan, Y.~Shi, and M.~Krstic, ``Neural operators for bypassing gain and control computations in pde backstepping,'' \emph{IEEE Transactions on Automatic Control}, vol.~69, no.~8, pp. 5310--5325, 2024.

\bibitem{bhan2026neuraloperatorpredictors}
L.~Bhan, ``Neural operator predictors for sampled measurements,'' 2026, gitHub repository. [Online]. Available: \href{https://github.com/lukebhan/NeuralOperatorPredictorsForSampledMeasurements} {github.com/lukebhan/NOPredictors}.

\bibitem{spong2006robot}
M.~W. Spong, S.~Hutchinson, and M.~Vidyasagar, \emph{Robot Modeling and Control}.\hskip 1em plus 0.5em minus 0.4em\relax Hoboken, NJ: Wiley, 2006.

\bibitem{Karafyllis2014}
\BIBentryALTinterwordspacing
I.~Karafyllis and M.~Krstic, ``Numerical schemes for nonlinear predictor feedback,'' \emph{Mathematics of Control, Signals, and Systems}, vol.~26, no.~4, pp. 519--546, Dec 2014. [Online]. Available: \url{https://doi.org/10.1007/s00498-014-0127-9}
\BIBentrySTDinterwordspacing

\bibitem{bhan2025neuraloperatorspredictorfeedback}
\BIBentryALTinterwordspacing
L.~Bhan, P.~Qin, M.~Krstic, and Y.~Shi, ``Neural operators for predictor feedback control of nonlinear delay systems,'' 2025. [Online]. Available: \url{https://arxiv.org/abs/2411.18964}
\BIBentrySTDinterwordspacing

\bibitem{kkk}
M.~Krstic, P.~V. Kokotovic, and I.~Kanellakopoulos, \emph{Nonlinear and Adaptive Control Design}, 1st~ed.\hskip 1em plus 0.5em minus 0.4em\relax USA: John Wiley \& Sons, Inc., 1995.

\end{thebibliography}
\bibliographystyle{IEEEtran.bst}

\end{document}